\documentclass[conference,a4paper]{IEEEtran}

\newtheorem{thm}{Theorem}

\newtheorem{lem}{Lemma}

\usepackage[final]{graphicx}
\usepackage[reqno]{amsmath}
\usepackage{amssymb}
\usepackage{subfig}
\usepackage{epstopdf}

\begin{document}

\sloppy

\title{Dynamic Interference Management}
\author{\IEEEauthorblockN{Aly El Gamal and Venugopal V.~Veeravalli}
 \IEEEauthorblockA{ECE Department and Coordinated Science Laboratory\\University of Illinois at Urbana-Champaign\\ Email: \{elgamal1,vvv\}@illinois.edu}}

\maketitle

\begin{abstract}
A linear interference network is considered. Long-term fluctuations (shadow fading) in the wireless channel can lead to any link being erased with probability $p$. Each receiver is interested in one unique message that can be available at $M$ transmitters. In a cellular downlink scenario, the case where $M=1$ reflects the cell association problem, and the case where $M>1$ reflects the problem of setting up the backhaul links for Coordinated Multi-Point (CoMP) transmission. In both cases, we analyze Degrees of Freedom (DoF) optimal schemes for the case of no erasures, and propose new schemes with better average DoF performance at high probabilities of erasure. For $M=1$, we characterize the average per user DoF, and identify the optimal assignment of messages to transmitters at each value of $p$. For general values of $M$, we show that there is no strategy for assigning messages to transmitters in large networks that is optimal for all values of $p$. 
\end{abstract}

\section{Introduction}

In~\cite{Mceliece-Stark-IT84}, the authors analyzed the average capacity for a point-to-point channel model where slow changes result in varying severity of noise. In this work, we apply a similar concept to interference networks by assuming that slowly changing deep fading conditions result in link erasures. We consider the linear interference network introduced by Wyner~\cite{Wyner}, with the consideration of two fading effects. Long-term fluctuations that result in link erasures over a complete block of time slots, and short-term fluctuations that allow us to assume that any specific joint realization for the non-zero channel coefficients, will take place with zero probability. We study the problem of associating receivers with transmitters and setting up the backhaul links for Coordinated Multi-Point (CoMP) transmission, in order to achieve the optimal average Degrees of Freedom (DoF). This problem was studied in~\cite{ElGamal-Annapureddy-Veeravalli-arXiv12} for the case of no erasures. Here, we extend the schemes in~\cite{ElGamal-Annapureddy-Veeravalli-arXiv12} to consider the occurrence of link erasures, and propose new schemes that lead to achieving better average DoF at high probabilities of erasure.

\section{System Model and Notation}\label{sec:systemmodel}
We use the standard model for the $K-$user interference channel with single-antenna transmitters and receivers,
\begin{equation}
Y_i(t) = \sum_{j=1}^{K} H_{i,j}(t) X_j(t) + Z_i(t),
\end{equation}
where $t$ is the time index, $X_j(t)$ is the transmitted signal of transmitter $j$, $Y_i(t)$ is the received signal at receiver $i$, $Z_i(t)$ is the zero mean unit variance Gaussian noise at receiver $i$, and $H_{i,j}(t)$ is the channel coefficient from transmitter $j$ to receiver $i$ over the time slot $t$. We remove the time index in the rest of the paper for brevity unless it is needed. For any set ${\cal A} \subseteq [K]$, we use the abbreviations $X_{\cal A}$, $Y_{\cal A}$, and $Z_{\cal A}$ to denote the sets $\left\{X_i, i\in {\cal A}\right\}$, $\left\{Y_i, i\in {\cal A}\right\}$, and $\left\{Z_i, i\in {\cal A}\right\}$, respectively. Finally, we use $[K]$ to denote the set $\{1,2,\ldots,K\}$.

\subsection{Channel Model}
Each transmitter can only be connected to its corresponding receiver as well as one following receiver, and the last transmitter can only be connected to its corresponding receiver. More precisely,

\begin{equation}\label{eq:channel}
H_{i,j} \text{ is identically } 0 \text { iff } i \notin \{j,j+1\},\forall i,j \in [K].
\end{equation}

In order to consider the effect of long-term fluctuations (shadowing), we assume that communication takes place over blocks of time slots, and let $p$ be the probability of block erasure. In each block, we assume that for each $j$, and each $i \in \{j,j+1\}$, $H_{i,j}=0$ with probability $p$. Moreover, short-term channel fluctuations allow us to assume that in each time slot, all non-zero channel coefficients are drawn independently from a continuous distribution. Finally, we assume that global channel state information is available at all transmitters and receivers. 
\subsection{Message Assignment}
For each $i \in [K]$, let $W_i$ be the message intended for receiver $i$, and ${\cal T}_i \subseteq [K]$ be the transmit set of receiver $i$, i.e., those transmitters with the knowledge of $W_i$. The transmitters in ${\cal T}_i$ cooperatively transmit the message $W_i$ to the receiver $i$. The messages $\{W_i\}$ are assumed to be independent of each other. The \emph{cooperation order} $M$ is defined to be the maximum transmit set size:
\begin{equation}\label{eq:coop_order}
M = \max_i |{\cal T}_i|.
\end{equation}

\subsection{Message Assignment Strategy}\label{sec:strategy}

A message assignment strategy is defined by a sequence of transmit sets $({\cal T}_{i,K}), i\in[K], K\in\{1,2,\ldots\}$. For each positive integer $K$ and $\forall i\in[K]$,  ${\cal T}_{i,K} \subseteq [K], |{\cal T}_{i,K}| \leq M$. We use message assignment strategies to define the transmit sets for a sequence of $K-$user channels. The $k^{\mathrm{th}}$ channel in the sequence has $k$ users, and the transmit sets for this channel are defined as follows. The transmit set of receiver $i$ in the $k^{\mathrm{th}}$ channel in the sequence is the transmit set ${\cal T}_{i,k}$ of the message assignment strategy. 
\subsection{Degrees of Freedom}
The average power constraint at each transmitter is $P$. In each block of time slots, the rates $R_i(P)$ are achievable if the decoding error probabilities of all messages can be simultaneously made arbitrarily small as the block length goes to infinity, and this holds for almost all realizations of non-zero channel coefficients. The sum capacity $\mathcal{C}_{\Sigma}(P)$ is the maximum value of the sum of the achievable rates. The total number of degrees of freedom ($\eta$) is defined as $\limsup_{P \rightarrow \infty}\frac{ C_{\Sigma}(P)}{\log P}$. For a probability of block erasure $p$, we let $\eta_p$ be the average value of $\eta$ over possible choices of non-zero channel coefficients.

For a $K$-user channel, we define $\eta_p(K,M)$ as the best achievable $\eta_p$ over all choices of transmit sets satisfying the cooperation order constraint in \eqref{eq:coop_order}. In order to simplify our analysis, we define the asymptotic average per user DoF $\tau_p(M)$ to measure how $\eta_p(K,M)$ scale with $K$,
\begin{equation}
\tau_p(M) = \lim_{K\rightarrow \infty} \frac{\eta_p(K,M)}{K}.
\end{equation}

We call a message assignment strategy \emph{optimal} for a given erasure probability $p$,  if there exists a sequence of coding schemes achieving $\tau_p(M)$ using the transmit sets defined by the message assignment strategy. A message assignment strategy is $\emph{universally optimal}$ if it is optimal for all values of $p$.

\section{Cell Association}
We first consider the case where each receiver can be served by only one transmitter. This reflects the problem of associating mobile users with cells in a cellular downlink scenario. We start by discussing orthogonal schemes (TDMA-based) for this problem, and then show that the proposed schemes are optimal. It will be useful in the rest of this section to view each realization of the network where some links are erased, as a series of subnetworks that do not interfere. We say that a set of $k$ users with successive indices $\{i,i+1,\ldots,i+k-1\}$ form a subnetwork if the following two conditions hold: The first condition is that $i=1$ or it is the case that message $W_{i-1}$ does not cause interference at $Y_i$, either because the direct link between the transmitter carrying $W_{i-1}$ and receiver $(i-1)$ is erased, or the transmitter carrying $W_{i-1}$ is not connected to the $i^{th}$ receiver. Secondly, $i+k-1=K$ or it is the case that message $W_{i+k-1}$ does not cause interference at $Y_{i+k}$, because the carrying transmitter is not connected to one of the receivers $(i+k-1)$ and $(i+k)$.

We say that the subnetwork is $\emph{atomic}$ if the transmitters carrying messages for users in the subnetwork have successive indices and for any transmitter $t$ carrying a message for a user in the subnetwork, and receiver $r$ such that $r \in \{t,t+1\}$ and $r\in\{i,i+1,\ldots,i+k-1\}$, the channel coefficient $H_{r,t} \neq 0$. 

For $i\in[K]$, let $N_i$ be the number of messages available at the $i^{th}$ transmitter, and let ${\bf N}^K=\left(N_1,N_2,\ldots,N_K\right)$. It is clear that the sequence ${\bf N}^K$ can be obtained from the transmit sets ${\cal T}_i, i\in[K]$; it is also true, as stated in the following lemma, that the converse holds. We borrow the notion of \emph{irreducible} message assignments from~\cite{ElGamal-Annapureddy-Veeravalli-arXiv12}. For $M=1$, an irreducible message assignment will have each message assigned to one of the two transmitters connected to its designated receiver.
\begin{lem}\label{lem:equiv}
For any irreducible message assignment where each message is assigned to exactly one transmitter, i.e., $|{\cal T}_i|=1, \forall i\in[K]$, the transmit sets ${\cal T}_i$, $i\in[K],$ are uniquely characterized by the sequence ${\bf N}^K$.
\end{lem}
\begin{proof}
Since each message can only be available at one transmitter, then this transmitter has to be connected to the designated receiver. More precisely, ${\cal T}_i \subset \{i-1,i\}, \forall i\in\{2,\ldots,K\}$, and ${\cal T}_1 = \{1\}$. It follows that each transmitter carries at most two messages and the first transmitter carries at least the message $W_1$, i.e., $N_i \in \{0,1,2\}, \forall i\in\{2,\ldots,K\}$, and $N_1 \in \{1,2\}$. Assume that $N_i=1, \forall i\in[K]$, then ${\cal T}_i=\{i\},\forall i\in[K]$. For the remaining case, we know that there exists $i \in \{2,\ldots,K\}$ such that $N_i=0$, since $\sum_{i=1}^{K} N_i = K$; we handle this case in the rest of the proof.

Let $x$ be the smallest index of a transmitter that carries no messages, i.e., $x = \min \{i: N_i=0\}$. We now show how to reconstruct the transmit sets ${\cal T}_i, i\in\{1,\ldots,x\}$ from the sequence $(N_1,N_2,\ldots,N_x)$. We note that ${\cal T}_i \in [x], \forall i\in[x]$, and since $N_x=0$, it follows that ${\cal T}_i \notin[x], \forall i\notin [x]$. It follows that $\sum_{i=1}^{x-1} N_i =x$. Since ${\cal T}_i \subset \{i-1,i\},\forall i\in\{2,\ldots,x\}$, we know that at most one transmitter in the first $x-1$ transmitters carries two messages. Since $\sum_{i=1}^{x-1} N_i =x$, and $N_i \in \{1,2\}, \forall i\in[x-1]$, it follows that there exists an index $y\in[x-1]$ such that $N_y = 2$, and $N_i=1, \forall i\in[x-1]\backslash\{y\}$. It is now clear that the $y^{th}$ transmitter carries messages $W_y$ and $W_{y+1}$, and each transmitter with an index $j\in\{y+1,\ldots,x-1\}$ is carrying message $W_{j+1}$, and each transmitter with an index $j \in \{1,\ldots,y\}$ is carrying message $W_j$. The transmit sets are then determined as follows. ${\cal T}_i=\{i\},\forall i\in[y]$ and ${\cal T}_i=\{i-1\},\forall i\in\{y+1,\ldots,x\}$.

We view the network as a series of subnetworks, where the last transmitter in each subnetwork is either inactive or the last transmitter in the network. If the last transmitter in a subnetwork is inactive, then the transmit sets in the subnetwork are determined in a similar fashion to the transmit sets ${\cal T}_i, i\in[x]$, in the above scenario. If the last transmitter in the subnetwork is the $K^{th}$ transmitter, and $N_K = 1$, then each message in this subnetwork is available at the transmitter with the same index.
\end{proof}

We use Lemma~\ref{lem:equiv} to describe message assignment strategies for large networks through repeating patterns of short ternary strings. Given a ternary string ${\bf S}=(S_1,\ldots,S_n)$ of fixed length $n$ such that $\sum_{i=1}^{n} S_i = n$, we define ${\bf N}^K$, $K \geq n$ as follows:
\begin{itemize}
\item $N_i=S_{i \text{ mod } n}$ if $\quad i\in\left\{1,\ldots,n\left\lfloor \frac{K}{n} \right\rfloor\right\}$,
\item $N_i=1$ if $i\in\left\{n\left\lfloor \frac{K}{n} \right\rfloor+1,\ldots,K\right\}$.
\end{itemize}

We now evaluate all possible message assignment strategies satisfying the cell association constraint using ternary strings through the above representation. We only restrict our attention to irreducible message assignments, and note that if there are two transmitters with indices $i$ and $j$ such that $i < j$ and each is carrying two messages, then there is a third transmitter with index $k$ such that $i < k < j$ that carries no messages. It follows that any string defining message assignment strategies that satisfy the cell association constraint, has to have one of the following forms:
\begin{itemize}
\item $S^{(1)}=(1)$,
\item $S^{(2)}=(2,1,1,\ldots,1,0)$,
\item $S^{(3)}=(1,1,\ldots,1,2,0)$,
\item $S^{(4)}=(1,1,\ldots,1,2,1,1,\ldots,1,0)$.
\end{itemize} 
We now introduce the three candidate message assignment strategies illustrated in Figure~\ref{fig:msgassignment}, and characterize the TDMA per user DoF achieved through each of them; we will show later that the optimal message assignment strategy at any value of $p$ is given by one of the three introduced strategies. We first consider the message assignment strategy defined by the string having the form $S^{(1)}=(1)$. Here, each message is available at the transmitter having the same index.

\begin{figure}
  \centering
\subfloat[]{\label{fig:highp}\includegraphics[height=0.1\textwidth]{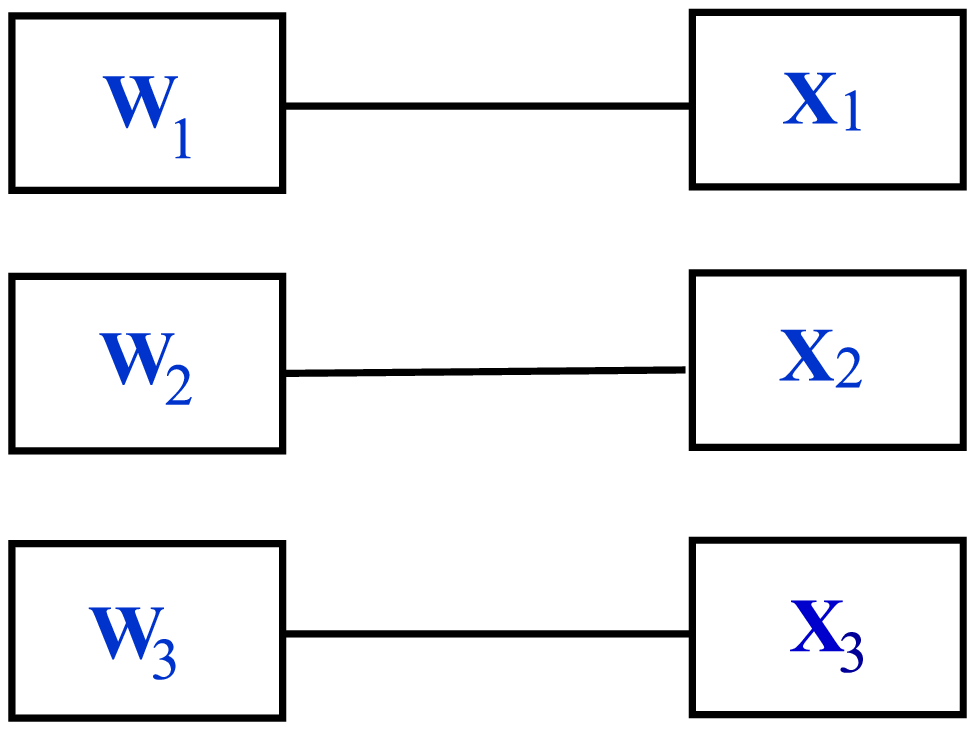}}                
\quad\quad\subfloat[]{\label{fig:lowp}\includegraphics[width=0.13\textwidth]{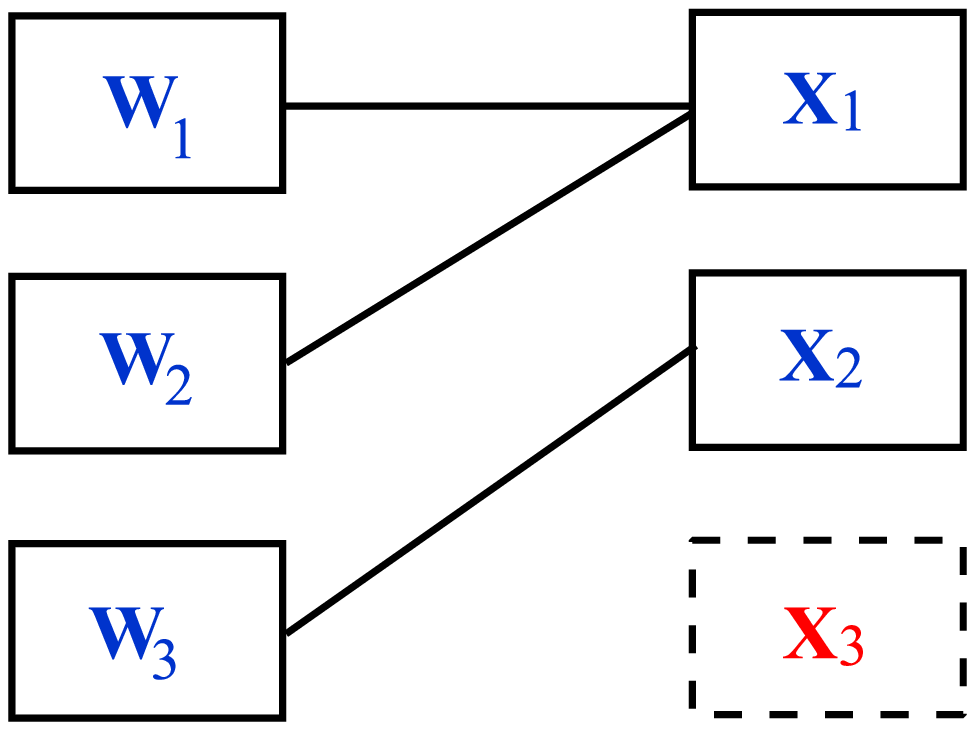}}
\quad\quad\subfloat[]{\label{fig:middlep}\includegraphics[width=0.13\textwidth]{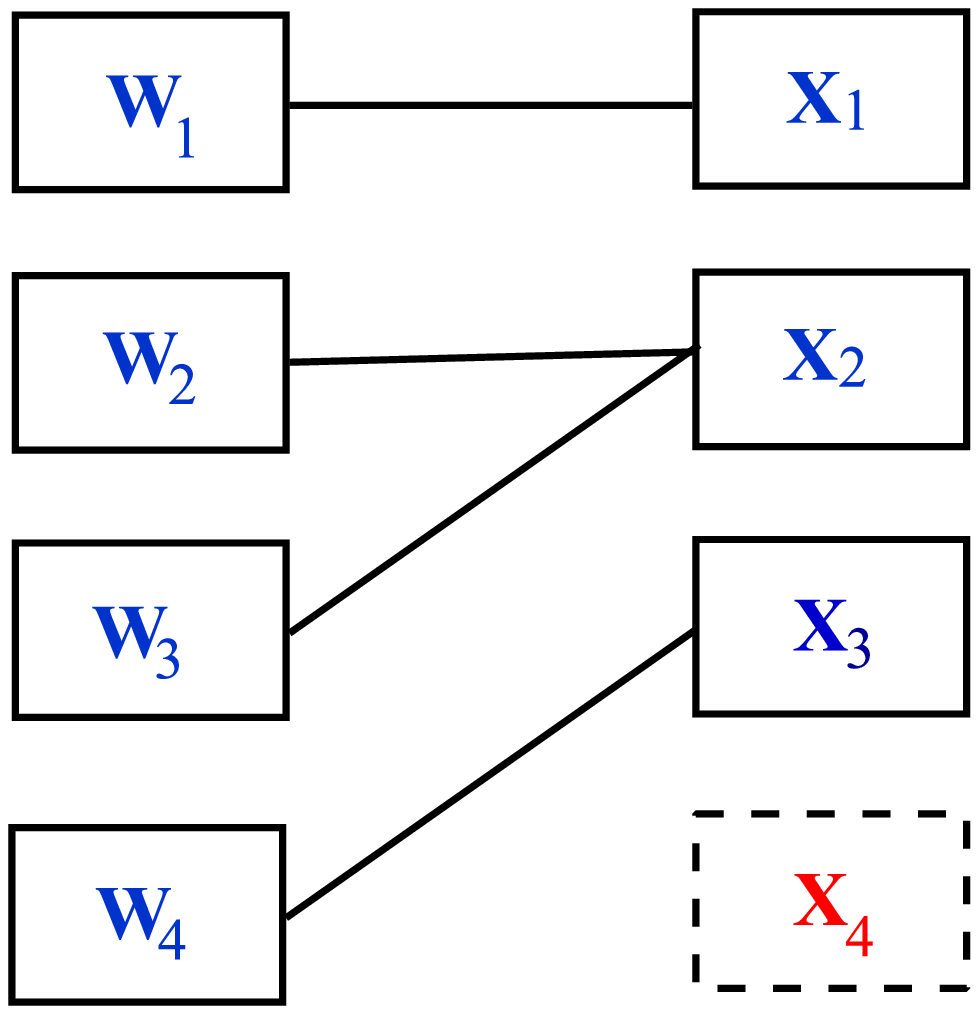}}
  \caption{The optimal message assignment strategies for the cell association problem. The red dashed boxes represent transmit signals that are inactive in all network realizations. The strategies in $(a)$, $(b)$, and $(c)$ are optimal at high, low, and middle values of the erasure probability $p$, respectively.}
  \label{fig:msgassignment}
\end{figure}
\begin{lem}\label{lem:highp}
Under the restriction to the message assignment strategy ${\cal T}_{i,K}=\{i\}, \forall K \in {\bf Z}^+, i\in[K],$ and orthogonal TDMA schemes, the average per user DoF is given by,
\begin{eqnarray}\label{eq:highp}
\tau_p^{(1)} &=& \frac{1}{2}\left(1-p+(1-p)\left(1-(1-p)^2\right)^2\right)\nonumber\\&&+\sum_{i=1}^{\infty} \frac{1}{2}\left(1-(1-p)^2\right)^2 (1-p)^{4i+1}.
\end{eqnarray}
\end{lem}
\begin{proof}
We will first explain a transmission scheme where $\frac{1}{2}\left(1-p+(1-p)\left(1-(1-p)^2\right)^2\right)$ DoF is achieved, and then modify it to show how to achieve $\tau_p^{(1)}$. For each user with and odd index $i$, message $W_i$ is transmitted whenever the channel coefficient $H_{i,i} \neq 0$; the rate achieved by these users contributes to the average per user DoF by $\frac{1}{2}(1-p)$. For each user with an even index $i$, message $W_i$ is transmitted whenever the following holds: $H_{i,i}\neq0$, $W_{i-1}$ does not cause interference at $Y_i$, and the transmission of $W_i$ will not disrupt the communication of $W_{i+1}$ to its designated receiver; we note that this happens if and only if $H_{i,i} \neq 0 \text { and } \left(H_{i-1,i-1}=0 \text{ or } H_{i,i-1}=0\right)$$\text{ and }$
$(H_{i+1,i}=0 \text{ or } H_{i+1,i+1}=0)$. It follows that the rate achieved by users with even indices contributes to the average per user DoF by $\frac{1}{2} (1-p)\left(1-(1-p)^2\right)^2$.

 We now show a modification of the above scheme to achieve $\tau_p^{(1)}$. As above, users with odd indices have priority, i.e., their messages are delivered whenever their direct links exist, and users with even indices deliver their messages whenever their direct links exist and the channel connectivity allows for avoiding conflict with priority users. However, we make an exception to the priority setting in atomic subnetworks consisting of an odd number of users, and the first and last users have even indices; in these subnetworks, one extra DoF is achieved by allowing users with even indices to have priority and deliver their messages. The resulting extra term in the average per user DoF is calculated as follows. Fixing a user with an even index, the probability that this user is the first user in a subnetwork consisting of an odd number of users in a large network is $\sum_{i=1}^{\infty}\left(1-(1-p)^2\right)^2 (1-p)^{4i+1}$; for each of these events, the sum DoF is increased by $1$, and hence the added term to the average per user DoF is equal to half this value, since every other user has an even index.

The optimality of the above scheme within the class of orthogonal TDMA-based schemes follows directly from~\cite[Theorem $1$]{Maleki-Jafar-arXiv13} for each realization of the network.
\end{proof}

We will show later that the above scheme is optimal at high erasure probabilities. In~\cite{ElGamal-Annapureddy-Veeravalli-arXiv12}, the optimal message assignment for the case of no erasures was characterized. The per user DoF was shown to be $\frac{2}{3}$, and was achieved by deactivating every third transmitter and achieving $1$ DoF for each transmitted message. We now consider the extension of this message assignment illustrated in Figure~\ref{fig:lowp}, which will be shown later to be optimal for low erasure probabilities. 
\begin{lem}\label{lem:lowp}
Under the restriction to the message assignment strategy defined by the string $S=(2,1,0)$, and orthogonal TDMA schemes, the average per user DoF is given by,
\begin{eqnarray}\label{eq:lowp}
\tau_p^{(2)} &=& \frac{2}{3}\left(1-p\right)+\frac{1}{3}p\left(1-p\right)\left(1-\left(1-p\right)^2 \right).
\end{eqnarray}
\end{lem}
\begin{proof}
For each user with an index $i$ such that $\left(i \text{ mod } 3 = 0\right)$ or $\left(i \text{ mod } 3=1\right)$, message $W_i$ is transmitted whenever the link between the transmitter carrying $W_i$ and the $i^{th}$ receiver is not erased; these users contribute to the average per user DoF by a factor of $\frac{2}{3}\left(1-p\right)$. For each user with an index $i$ such that $\left(i \text{ mod } 3=2\right)$, message $W_i$ is transmitted through $X_{i-1}$ whenever the following holds: $H_{i,i-1} \neq 0$, message $W_{i-1}$ is not transmitted because $H_{i-1,i-1}=0$, and the transmission of $W_i$ will not be disrupted by the communication of $W_{i+1}$ through $X_i$ because $\left(H_{i,i}=0\right) \text{ or } \left(H_{i+1,i}=0\right)$; these users contribute to the average per user DoF by a factor of $\frac{1}{3}p\left(1-p\right)\left(1-\left(1-p\right)^2\right)$. Using the considered message assignment strategy, the TDMA optimality of this scheme follows from~\cite[Theorem $1$]{Maleki-Jafar-arXiv13} for each network realization.
\end{proof}

We now consider the message assignment strategy illustrated in Figure~\ref{fig:middlep}. We will show later that this strategy is optimal for a middle regime of erasure probabilities. 
\begin{lem}\label{lem:middlep}
Under the restriction to the message assignment strategy defined by the string $S=(1,2,1,0)$, and orthogonal TDMA schemes, the average per user DoF is given by,
\begin{eqnarray}\label{eq:middlep}
\tau_p^{(3)} &=& \frac{1}{2}\left(1-p\right)\nonumber\\&&+\frac{1}{4}\left(1-p\right)\left(1-\left(1-p\right)^2 \right)\left(1+p+\left(1-p\right)^3\right).\nonumber\\
\end{eqnarray}
\end{lem}
\begin{proof}
As in the proof of Lemma~\ref{lem:highp}, we first explain a transmission scheme achieving part of the desired rate, and then modify it to show how the extra term can be achieved. Let each message with an odd index be delivered whenever the link between the transmitter carrying the message and the designated receiver is not erased; these users contribute to the average per user DoF by a factor of $\frac{1}{2} \left(1-p\right)$. For each user with an even index $i$, if $i \text{ mod } 4=2$, then $W_i$ is transmitted through $X_i$ whenever the following holds: $H_{i,i} \neq 0$, message $W_{i+1}$ is not transmitted through $X_i$ because $H_{i+1,i} =0$, and the transmission of $W_i$ can will not be disrupted by the communication of $W_{i-1}$ through $X_{i-1}$ because either $H_{i,i-1}=0$ or $H_{i-1,i-1} =0$; these users contribute to the average per user DoF by a factor of $\frac{1}{4}p\left(1-p\right)\left(1-\left(1-p\right)^2\right)$. For each user with an even index $i$ such that $i$ is a multiple of $4$, $W_i$ is transmitted through $X_{i-1}$ whenever $H_{i,i-1}\neq 0$, and the transmission of $W_i$ will not disrupt the communication of $W_{i-1}$ through $X_{i-2}$ because either $H_{i-1,i-1}=0$ or $H_{i-1,i-2}=0$; these users contribute to the average per user DoF by a factor of $\frac{1}{4}\left(1-p\right)\left(1-\left(1-p\right)^2\right)$. 

We now modify the above scheme to show how $\tau_p^{(3)}$ can be achieved. Since the $i^{th}$ transmitter is inactive for every $i$ that is a multiple of $4$, users $\{i-3,i-2,i-1,i\}$ are separated from the rest of the network for every $i$ that is a multiple of $4$, i.e., these users form a subnetwork. We explain the modification for the first four users, and it will be clear how to apply a similar modification for every following set of four users. Consider the event where message $W_1$ does not cause interference at $Y_2$, because either $H_{1,1}=0$ or $H_{2,1}=0$, and it is the case that $H_{2,2}\neq 0$, $H_{3,2} \neq 0$, $H_{3,3} \neq 0$, and $H_{4,3} \neq 0$; this is the event that users $\{2,3,4\}$ form an atomic subnetwork, and it happens with probability $\left(1-\left(1-p\right)^2\right)\left(1-p\right)^4$. In this case, we let messages $W_2$ and $W_4$ have priority instead of message $W_3$, and hence the sum DoF for messages $\{W_1,W_2,W_3,W_4\}$ is increased by $1$. It follows that an extra term of $\frac{1}{4}\left(1-\left(1-p\right)^2\right)\left(1-p\right)^4$ is added to the average per user DoF. 

The TDMA optimality of the illustrated scheme follows from~\cite[Theorem $1$]{Maleki-Jafar-arXiv13} for each network realization.
\end{proof}

In Figure~\ref{fig:monenorm}, we plot the values of $\frac{\tau_p^{(1)}}{1-p}$, $\frac{\tau_p^{(2)}}{1-p}$, and $\frac{\tau_p^{(3)}}{1-p}$, and note that $\max \left\{\tau_p^{(1)},\tau_p^{(2)},\tau_p^{(3)}\right\}$ equals $\tau_p^{(1)}$ at high probabilities of erasure, and equals $\tau_p^{(2)}$ at low probabilities of erasure, and equals $\tau_p^{(3)}$ in a middle regime. 
\begin{figure}[htb]
\centering
\includegraphics[width=1\columnwidth]{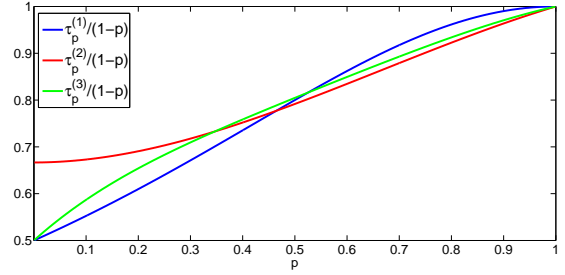}
\caption{The average per user DoF achieved through the strategies in Lemmas~\ref{lem:highp},~\ref{lem:lowp}, and~\ref{lem:middlep}, normalized by $(1-p)$.}
\label{fig:monenorm}
\end{figure} 

We now show that under the restriction to TDMA schemes, one of the message assignment strategies illustrated in Lemmas~\ref{lem:highp}, \ref{lem:lowp}, and \ref{lem:middlep} is optimal at any value of $p$. 
\begin{thm}\label{thm:tdma}
For a given erasure probability $p$, let $\tau_p^{(\text{TDMA})}$ be the average per user DoF under the restriction to orthogonal TDMA schemes, then at any value $0 \leq p \leq 1$ the following holds,
\begin{equation}\label{eq:tdma}
\tau_p^{(TDMA)}=\max \left\{\tau_p^{(1)},\tau_p^{(2)},\tau_p^{(3)}\right\},
\end{equation}
where $\tau_p^{(1)}$, $\tau_p^{(2)}$, and $\tau_p^{(3)}$ are given in~\eqref{eq:highp},~\eqref{eq:lowp}, and~\eqref{eq:middlep}, respectively.
\end{thm}
\begin{proof}
The inner bound follows from Lemmas~\ref{lem:highp},~\ref{lem:lowp}, and~\ref{lem:middlep}. In order to prove the converse, we need to consider all irreducible message assignment strategies where each message is assigned to a single transmitter. 
We know from Lemma~\ref{lem:highp} that the TDMA average per user DoF achieved through the strategy defined by the string of all ones having the form $S^{(1)}=(1)$ equals $\tau_p^{(1)}$, and hence the upper bound holds in this case. 

We now show that the TDMA average per user DoF achieved through strategies defined by strings of the form $S^{(2)}=\left(2,1,\ldots,1,0\right)$ is upper bounded by a convex combination of $\tau_p^{(1)}$ and $\tau_p^{(2)}$, and hence, is upper bounded by $\max \left\{\tau_p^{(1)},\tau_p^{(2)}\right\}$. The considered message assignment strategy splits each network into subnetworks consisting of a transmitter carrying two messages followed by a number of transmitters, each is carrying one message, and the last transmitter in the subnetwork carries no messages. We first consider the case where the number of transmitters carrying single messages is odd. We consider the simple scenario of the message assignment strategy defined by the string $(2,1,1,1,0)$, and then the proof will be clear for strategies defined by strings of the form $\left(2,1,1,\ldots,1,0\right)$ that have an arbitrary odd number of ones. In this case, it suffices to show that the average per user DoF in the first subnetwork is upper bounded by a convex combination of $\tau_p^{(1)}$ and $\tau_p^{(2)}$. The first subnetwork consists of the first five users; $W_1$ and $W_2$ can be transmitted through $X_1$. $W_3$, $W_4$ and $W_5$ can be transmitted through $X_2$, $X_3$, and $X_4$, respectively, and the transmit signal $X_5$ is inactive. 

We now explain the optimal TDMA scheme for the considered subnetwork. We first explain a simple scheme and then modify it to get the optimal scheme. Each of the messages $W_1$, $W_3$, and $W_5$ is delivered whenever the direct link between its carrying transmitter and its designated receiver is not erased. Message $W_2$ is delivered whenever message $W_1$ is not transmitted, and message $W_3$ is not causing interference at $Y_2$. Message $W_4$ is transmitted whenever $W_5$ is not causing interference at $Y_4$, and the transmission of $W_4$ through $X_3$ will not disrupt the communication of $W_3$. We now explain the modification; if there is an atomic subnetwork consisting of users $\{2,3,4\}$, then we switch the priority setting within this subnetwork, and messages $W_2$ and $W_4$ will be delivered instead of message $W_3$. The TDMA optimality of this scheme for each realization of the network follows from~\cite[Theorem $1$]{Maleki-Jafar-arXiv13}. Now, we note that the average sum DoF for messages $\{W_1,\ldots,W_5\}$ is equal to their sum DoF in the original scheme plus an extra term due to the modification. The average sum DoF for messages $\{W_1,W_2,W_5\}$ in the original scheme equals $3\tau_p^{(2)}$, and the sum of the average sum DoF for messages $\{W_3,W_4\}$ and the extra term is upper bounded by $2 \tau_p^{(1)}$. It follows that the average per user DoF is upper bounded by $\frac{2}{5} \tau_p^{(1)} + \frac{3}{5} \tau_p^{(2)}$. The proof can be generalized to show that the average TDMA per user DoF for message assignment strategies defined by strings of the form $S^{(2)}$ with an odd number of ones $n$, is upper bounded by $\frac{n-1}{n+2} \tau_p^{(1)} + \frac{3}{n+2} \tau_p^{(2)}$.

For message assignment strategies defined by a string of the form $S^{(2)}$ with an even number of ones $n$, it can be shown in a similar fashion as above that the TDMA average per user DoF is upper bounded by $\frac{n}{n+2} \tau_p^{(1)} + \frac{2}{n+2} \tau_p^{(2)}$. Also, for strategies defined by a string of the form $S^{(3)}=\left(1,1,\ldots,1,2,0\right)$ with a number of ones $n$, the TDMA average per user DoF is the same as that of a strategy defined by a string of the form $S^{(2)}$ with the same number of ones, and hence, is upper bounded by a convex combination of $\tau_p^{(1)}$ and $\tau_p^{(2)}$. Finally, for strategies defined by a string of the form $S^{(4)}=\left(1,1,\ldots,1,2,1,1,\ldots,1,0\right)$ with a number of ones $n$, it can be shown in a similar fashion as above that the average per user DoF is upper bounded by $\frac{n-2}{n+2}\tau_p^{(1)}+\frac{4}{n+2}\tau_p^{(3)}$.
\end{proof}

We now characterize the average per user DoF for the cell association problem by proving that TDMA schemes are optimal for any candidate message assignment strategy. In order to prove an information theoretic upper bound on the per user DoF for each network realization, we use Lemma $4$ from~\cite{ElGamal-Annapureddy-Veeravalli-arXiv12}, which we restate below. For any set of receiver indices ${\cal A} \subseteq [K]$, define $U_{\cal A}$ as the set of indices of transmitters that exclusively carry the messages for the receivers in ${\cal A}$, and the complement set is $\bar{U}_{\cal A}$. More precisely, $U_{\cal A} = [K]\backslash\cup_{i \notin {\cal A}} {\cal T}_i$.
\begin{lem} [\cite{ElGamal-Annapureddy-Veeravalli-arXiv12}] \label{lem:dofouterbound}
If there exists a set ${\cal A}\subseteq [K]$, a function $f_1$, and a function $f_2$ whose definition does not depend on the transmit power constraint $P$, and $f_1\left(Y_{\cal A},X_{U_{\cal A}}\right)=X_{\bar{U}_{\cal A}}+f_2(Z_{\cal A})$, then the sum DoF $\eta \leq |{\cal A}|$.
\end{lem}
\begin{thm}\label{thm:mone}
The average per user DoF for the cell association problem is given by,

\begin{equation}\label{eq:tauone}
\tau_p\left(M=1\right)=\tau_p^{(TDMA)}=\max \left\{\tau_p^{(1)},\tau_p^{(2)},\tau_p^{(3)}\right\},
\end{equation}
where $\tau_p^{(1)}$, $\tau_p^{(2)}$, and $\tau_p^{(3)}$ are given in~\eqref{eq:highp},~\eqref{eq:lowp}, and~\eqref{eq:middlep}, respectively.
\end{thm}
\begin{proof}
In order to prove the statement, we need to show that $\tau_p(M=1) \leq \tau_p^{(TDMA)}$; we do so by using Lemma~\ref{lem:dofouterbound} to show that for any irreducible message assignment strategy satisfying the cell association constraint, and any network realization, the asymptotic per user DoF is given by that achieved through the optimal TDMA scheme.

Consider message assignment strategies defined by strings having one of the forms $S^{(1)}=(1)$, $S^{(2)}=\left(2,1,1,\ldots,1,0\right)$, and $S^{(3)}=\left(1,1,\ldots,1,2,0\right)$. We view each network realization as a series of atomic subnetworks, and show that for each atomic subnetwork, the sum DoF is achieved by the optimal TDMA scheme. For an atomic subnetwork consisting of a number of users $n$, we note that $\left\lfloor\frac{n+1}{2}\right\rfloor$ users are active in the optimal TDMA scheme; we now show in this case using Lemma~\ref{lem:dofouterbound} that the sum DoF for users in the subnetwork is bounded by $\left\lfloor\frac{n+1}{2}\right\rfloor$. Let the users in the atomic subnetwork have the indices $\{i,i+1,\ldots,i+n-1\}$, then we use Lemma~\ref{lem:dofouterbound} with the set ${\cal A}=\left\{i+2j: j\in\left\{0,1,2,\ldots,\left\lfloor\frac{n-1}{2}\right\rfloor\right\}\right\}$, except the cases of message assignment strategies defined by strings having one of the forms $S^{(1)}=(1)$ and $S^{(3)}=\left(1,1,\ldots,1,2,0\right)$ with an even number of ones, where we use the set ${\cal A}=\left\{i+1+2j: j\in\left\{0,1,2,\ldots,\frac{n-2}{2}\right\}\right\}$. We now note that each transmitter that carries a message for a user in the atomic subnetwork and has an index in $\bar{U}_{\cal A}$, is connected to a receiver in ${\cal A}$, and this receiver is connected to one more transmitter with an index in $U_{\cal A}$, and hence, the missing transmit signals $X_{\bar{U}_{\cal A}}$ can be recovered from $Y_{\cal A}-Z_{\cal A}$ and $X_{U_{\cal A}}$. The condition in the statement of Lemma~\ref{lem:dofouterbound} is then satisfied; allowing us to prove that the sum DoF for users in the atomic subnetwork is upper bounded by $|{\cal A}|=\left\lfloor\frac{n+1}{2}\right\rfloor$.

The proof is similar for message assignment strategies defined by strings that have the form $S^{(4)}=\left\{1,1,\ldots,1,2,1,1,\ldots,1,0\right\}$. However, there is a difference in selecting the set ${\cal A}$ for atomic subnetworks consisting of users with indices $\{i,i+1,\ldots,i+x,i+x+1,\ldots,i+n-1\}$, where $1 \leq x \leq n-2$, and messages $W_{i+x}$ and $W_{i+x+1}$ are both available at transmitter $i+x$. In this case, we apply Lemma~\ref{lem:dofouterbound} with the set ${\cal A}$ defined as above, but including indices $\{i+x,i+x+1\}$ and excluding indices $\{i+x-1,i+x+2\}$. It can be seen that the condition in Lemma~\ref{lem:dofouterbound} will be satisfied in this case, and the proved upper bound on the sum DoF for each atomic subnetwork, is achievable through TDMA.
\end{proof}
\begin{figure}[htb]
\centering
\includegraphics[width=1\columnwidth]{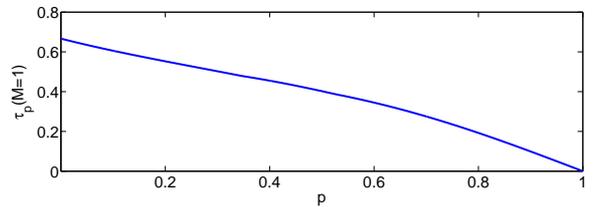}
\caption{The average per user DoF for the cell association problem}
\label{fig:monemax}
\end{figure} 

In Figure~\ref{fig:monemax}, we plot $\tau_p(M=1)$ at each value of $p$. The result of Theorem~\ref{thm:mone} implies that the message assignment strategies considered in Lemmas~\ref{lem:highp},~\ref{lem:lowp},~\ref{lem:middlep} are optimal at high, low, and middle values of the erasure probability $p$, respectively. We note that in densely connected networks at a low probability of erasrue, the \emph{interference-aware} message assignment strategy in Figure~\ref{fig:lowp} is optimal; through this assignment, the maximum number of interference free communication links can be created for the case of no erasures. On the other hand, the linear nature of the channel connectivity does not affect the choice of optimal message assignment at high probability of ersure. As the effect of interference diminishes at high probability of erasure, assigning each message to a unique transmitter, as in the strategy in Figure~\ref{fig:highp}, becomes the only criterion of optimality. At middle values of $p$, the message assignment strategy in Figure~\ref{fig:middlep} is optimal; in this assignment, the network is split into four user subnetworks. In the first subnetwork, the assignment is optimal as the maximum number of interference free communication links can be created for the two events where there is an atomic subnetwork consisting of users $\{1,2,3\}$ or users $\{2,3,4\}$.

\section{Coordinated Multi-Point Transmission}\label{sec:comp}
We have shown that there is no message assignment strategy for the cell association problem that is optimal for all values of $p$. We show in this section that this statement is true even for the case where each message can be available at more than one transmitter ($M>1$). Recall that for a given value of $M$, we say that a message assignment strategy is universally optimal if it can be used to achieve $\tau_p(M)$ for all values of $p$.
\begin{thm}\label{thm:comp}
For any value of the cooperation constraint $M \in {\bf Z}^+$, there does not exist a universally optimal message assignment strategy
\end{thm}
\begin{proof}
The proof follows from Theorem~\ref{thm:mone} for the case where $M=1$. We show that for any value of $M>1$,  any message assignment strategy that enables the achievability of $\tau_p(M)$ at high probabilities of erasure, is not optimal for the case of no erasures, i.e., cannot be used to achieve $\tau_p(M)$ for $p=0$. For any message assignment strategy, consider the value of $\lim_{p \rightarrow 1} \frac{\tau_p(M)}{1-p}$ and note this value equals the average number of transmitters in a transmit set that can be connected to the designated receiver. More precisely,
\begin{equation}\label{eq:cond}
\lim_{p \rightarrow 1} \frac{\tau_p(M)}{1-p}=\frac{\sum_{i=1}^{K}|{\cal T}_i \cap \{i-1,i\}|}{K},
\end{equation}
where ${\cal T}_i$ in~\eqref{eq:cond} corresponds to an optimal message assignment strategy at high probabilities of erasure. It follows that there exists a value $0 < \bar{p} < 1$ such that for any message assignment strategy that enables the achievability of $\tau_p(M)$ for $p \geq \bar{p}$, almost all messages are assigned to the two transmitters that can be connected to the designated receiver, i.e., if we let $S_K=\left\{i: {\cal T}_{i,K} = \left\{i-1,i\right\}\right\}$, then $\lim_{K \rightarrow \infty} \frac{|S_K|}{K} = 1$. 

We recall from~\cite{ElGamal-Annapureddy-Veeravalli-arXiv12} that for the case of no erasures, the average per user DoF equals $\frac{2M}{2M+1}$. We also note that following the same footsteps as in the proof of~\cite[Theorem $7$]{ElGamal-Annapureddy-Veeravalli-arXiv12}, we can show that for any message assignment strategy such that $\lim_{K \rightarrow \infty} \frac{|S_K|}{K} = 1$, the per user DoF for the case of no erasures is upper bounded by $\frac{2M-2}{2M-1}$; we do so by using Lemma~\ref{lem:dofouterbound} for each $K-$user channel with the set ${\cal A}$ defined such that the complement set $\bar{\cal A}=\{i:i\in[K], i=(2M-1)(j-1)+M, j\in{\bf Z}^+\}$.
\end{proof}

The condition of optimality identified in the proof of Theorem~\ref{thm:comp} for message assignment strategies at high probabilities of erasure suggest a new role for cooperation in dynamic interference networks. The availability of a message at more than one transmitter may not only be used to cancel its interference at other receivers, but to increase the chances of connecting the message to its designated receiver. This new role leads to three effects at high erasure probability.  The achieved DoF in the considered linear interference network becomes larger than that of $K$ parallel channels, in particular, $\lim_{p \rightarrow 1} \frac{\tau_p(M>1)}{1-p} = 2$. Secondly, as the effect of interference diminishes at high probabilities of erasures, all messages can simply be assigned to the two transmitters that may be connected to their designated receiver, and a simple interference avoidance scheme can be used in each network realization, as we show below in the scheme of Theorem~\ref{thm:mtwoic}. It follows that channel state information is no longer needed at transmitters, and only information about the slow changes in the network topology is needed to achieve the optimal average DoF. Finally, unlike the optimal scheme of~\cite[Theorem $4$]{ElGamal-Annapureddy-Veeravalli-arXiv12} for the case of no erasures, where some transmitters are always inactive, achieving the optimal DoF at high probabilities of erasure requires all transmitters to be used in at least one network realization.

We now restrict our attention to the case where $M=2$. Here, each message can be available at two transmitters, and transmitted jointly by both of them. We study two message assignment strategies that are optimal in the limits of $p \rightarrow 0$ and $p \rightarrow 1$, and derive inner bounds on the average per user DoF $\tau_p(M=2)$ based on the considered strategies. In~\cite{ElGamal-Annapureddy-Veeravalli-arXiv12}, the message assignment of Figure~\ref{fig:mtwojone} was shown to be DoF optimal for the case of no erasures ($p=0$). The network is split into subnetworks, each with five consecutive users. The last transmitter of each subnetwork is deactivated to eliminate inter-subnetwork interfeerence. In the first subnetwork, message $W_3$ is not transmitted, and each other message is received without interference at its designated receiver. Note that the transmit beams for messages $W_1$ and $W_5$ contributing to the transmit signals $X_2$ and $X_5$, respectively, are designed to cancel the interference at receivers $Y_2$ and $Y_4$, respectively. An analog scheme is used in each following subnetwork. The value of $\tau_p(M=2)$ is thus $\frac{4}{5}$ for the case where $p=0$. In order to prove the following result, we extend the message assignment of Figure~\ref{fig:mtwojone} to consider the possible presence of block erasures. 

\begin{thm}\label{thm:mtwoicaware}
For $M=2$, the following average per user DoF is achievable,
\begin{equation}\label{eq:mtwoicaware}
\tau_p(M=2) \geq \frac{2}{5} \left(1-p\right)\left(2 + A.p\right),
\end{equation}
where,
\begin{equation}
A = p + 1 - \left(\left(1-p\right)^2 \left(1-p\left(1-p\right)\right)\right) - \frac{1}{2} p(1-p),
\end{equation}
and is asymptotically optimal as $p \rightarrow 0$.
\end{thm}
\begin{proof}
We know from~\cite{ElGamal-Annapureddy-Veeravalli-arXiv12} that $\lim_{p \rightarrow 0} \tau_p(2)=\frac{4}{5}$, and hence, it suffices to show that the inner bound in~\eqref{eq:mtwoicaware} is valid. For each $i \in [K]$. message $W_i$ is assigned as follows,

\vspace{5 mm}
${\cal T}_{i}=
\begin{cases}
\{i,i+1\}, \quad &\text{ if } i \equiv 1 \text{ mod } 5\\
\{i-1,i-2\}, \quad &\text{ if } i \equiv 0 \text{ mod } 5\\
\{i-1,i\}, \quad &\text{ otherwise },
\end{cases}$
\vspace{5 mm}

We illustrate this message assignment in Figure~\ref{fig:mtwojonenew}. We note that the transmit signals $\{X_i: i \equiv 0 \text{ mod } 5\}$ are inactive, and hence, we split the network into five user subnetworks with no interference between successive subnetworks. We explain the transmission scheme in the first subnetwork and note that a similar scheme applies to each following subnetwork. In the proposed transmission scheme, any receiver is either inactive or receives its desired message without interference, and any transmitter will not transmit more than one message for any network realization. It follows that 1 DoF is achieved for each message that is transmitted.

 Messages $W_1$, $W_2$, $W_4$, and $W_5$ are transmitted through $X_1$, $X_2$, $X_3$, and $X_4$, respectively, whenever the coefficients $H_{1,1}\neq 0$, $H_{2,2}\neq 0$, $H_{4,3}\neq 0$, and $H_{5,4}\neq 0$, respectively. Note that the transmit beam for message $W_1$ contributing to $X_2$ can be designed to cancel its interference at $Y_2$. Similarly, the interference caused by $W_5$ at $Y_4$ can be cancelled through $X_3$. It follows that $(1-p)$ DoF is achieved for each of $\{W_1,W_2,W_4,W_5\}$, and hence, $\tau_p(2) \geq \frac{4}{5} (1-p)$. Also, message $W_2$ is transmitted through $X_1$ if it cannot be transmitted through $X_2$ and message $W_1$ is not transmitted through $X_1$. More precisely, message $W_2$ is transmitted through $X_1$ if $H_{2,2}=0$ and $H_{2,1} \neq 0$ and $H_{1,1}=0$, thereby achieving an extra $p^2 (1-p)$ DoF. Similarly, message $W_4$ can be transmitted through $X_4$ if $H_{4,3}=0$ and $H_{4,4} \neq 0$ and $H_{5,4} = 0$. It follows that,
\begin{equation}\label{eq:stepone}
\tau_p(2) \geq \frac{4}{5} (1-p) + \frac{2}{5} p^2 (1-p)
\end{equation}

Finally, message $W_3$ will be transmitted through $X_3$ if message $W_4$ is not transmitted through $X_3$, and message $W_2$ is not causing interference at $Y_3$. Message $W_4$ is not transmitted through $X_3$ whenever the coefficient $H_{4,3}=0$, and message $W_2$ does not cause interference at $Y_3$ whenever the coefficient $H_{2,2} = 0$ or the coefficient $H_{3,2} = 0$ or $W_2$ can be transmitted through $X_1$. More precisely, message $W_3$ is transmitted through $X_3$ if and only if all the following is true:
\begin{itemize}
\item $H_{3,3} \neq 0$, and $H_{4,3} = 0$
\item $H_{2,2} = 0$, or $H_{3,2} = 0$, or it is the case that $H_{1,1} = 0$ and $H_{2,1} \neq 0$.
\end{itemize} 
It follows that $f(p)$ DoF is achieved for message $W_3$, where,
\begin{equation}
f(p)=p\left(1-p\right)\left(1 - \left(\left(1-p\right)^2 \left(1-p\left(1-p\right)\right)\right)\right).
\end{equation}
Similary, $W_3$ can be transmitted through $X_2$ if and only if message $W_2$ is not transmitted through $X_2$ and message $W_4$ is either not transmitted or can be transmitted without causing interference at $Y_3$, i.e., if and only if all the following is true:
\begin{itemize}
\item $H_{3,2} \neq 0$, and $H_{2,2} = 0$
\item $H_{4,3} = 0$, or $H_{3,3} = 0$, or it is the case that $H_{5,4}=0$ and $H_{4,4} \neq 0$. 
\end{itemize} 
The above conditions are satisfied with probability $f(p)$. Since we have counted twice the event that $H_{3,3} \neq 0$ and $H_{4,3} =0$ and $H_{3,2} \neq 0$ and $H_{2,2} = 0$, it follows that $2f(p)-p^2(1-p)^2$ DoF is achieved for $W_3$. Summing the DoF achieved for other messages in~\eqref{eq:stepone}, we conclude that,
\begin{equation}
\tau_p(2) \geq \frac{4}{5} (1-p) + \frac{2}{5} p^2 (1-p) + \frac{1}{5} \left(2f(p)-p^2(1-p)^2 \right),
\end{equation}
which is the same inequality as in~\eqref{eq:mtwoicaware}.
\end{proof}

\begin{figure}
  \centering
  
\subfloat[]{\label{fig:mtwojone}\includegraphics[height=0.16\textwidth]{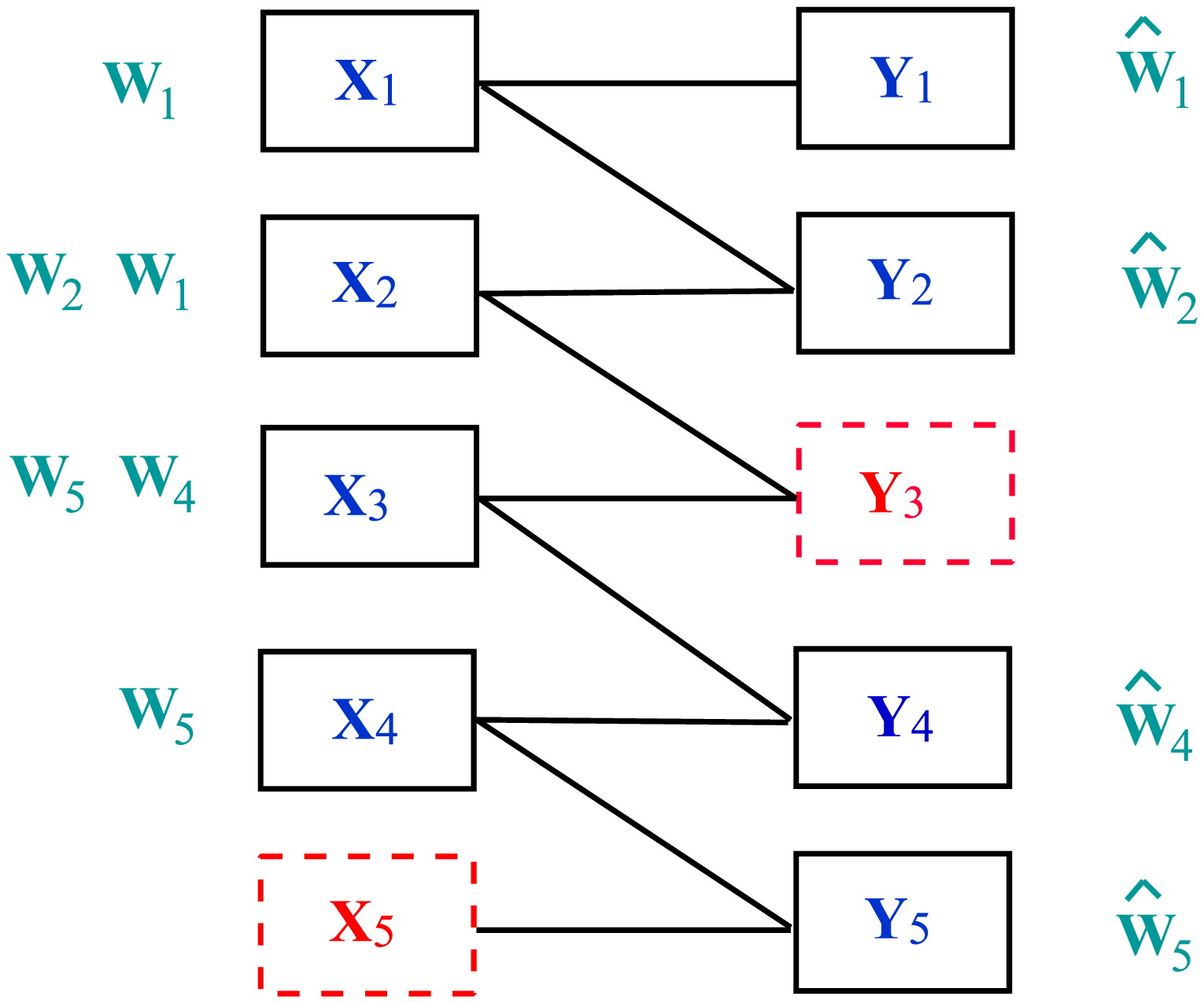}}                
\quad\quad\quad\quad\subfloat[]{\label{fig:mtwojonenew}\includegraphics[width=0.21\textwidth]{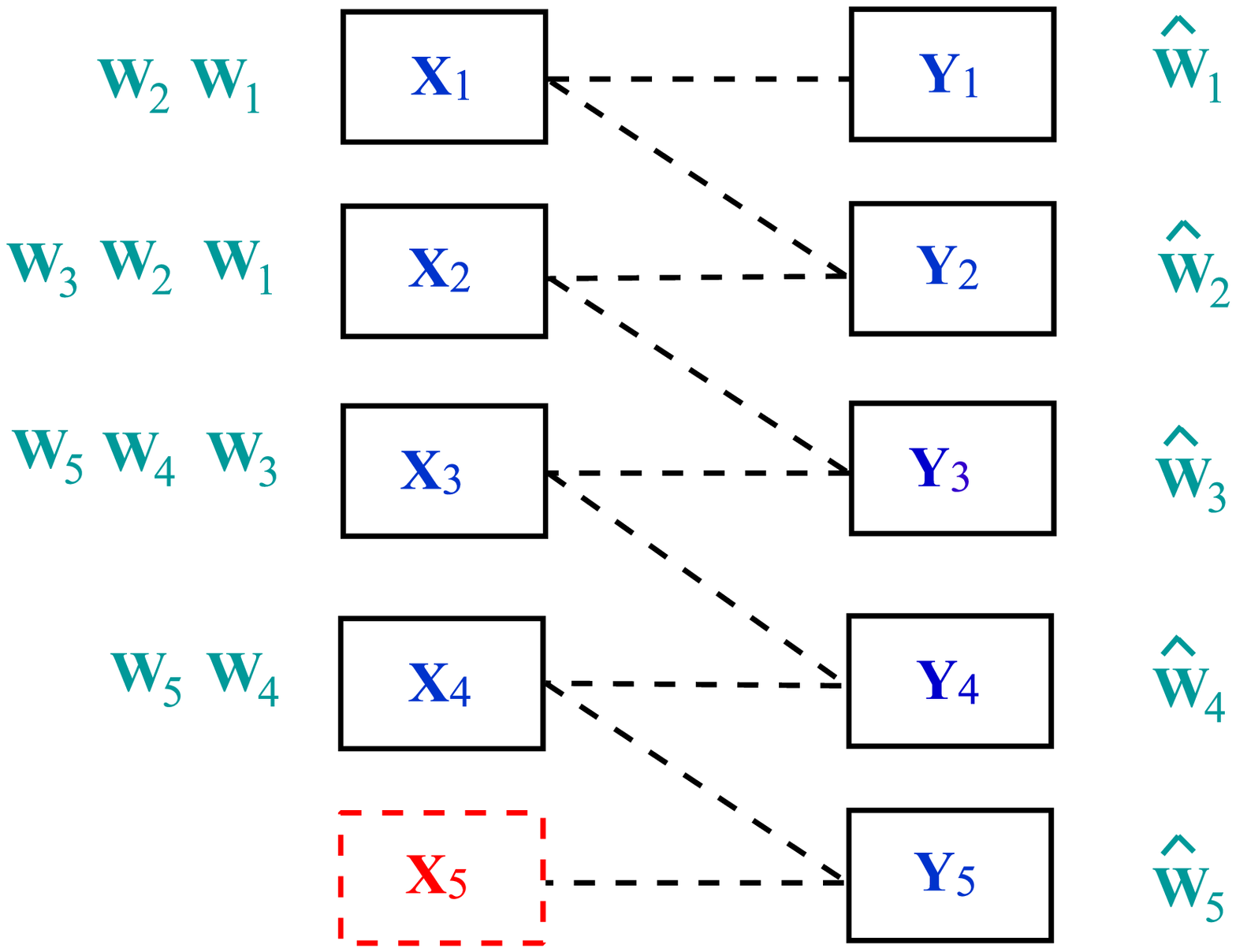}}
  \caption{The message assignment in ($a$) is optimal for a linear network with no erasures ($p=0$). We extend this message assignment in ($b$) to consider non-zero erasure probabilities. In both figures, the red dashed boxes correspond to inactive signals.}
  \label{fig:mtwoicaware}
\end{figure}

Although the scheme of Theorem~\ref{thm:mtwoicaware} is optimal for the case of no erasures ($p=0$), we know from Theorem~\ref{thm:comp} that better schemes exist at high erasure probabilities. Since in each five user subnet in the scheme of Theorem~\ref{thm:mtwoicaware}, only three users have their messages assigned to the two transmitters that can be connected to their receivers, and two users have only one of these transmitters carrying their messages, we get the asymptotic limit of  $\frac{8}{5}$ for the achieved average per user DoF normalized by $(1-p)$ as $p \rightarrow 1$. This leads us to consider an alternative message assignment where the two transmitters carrying each message $i$ are the two transmitters $\left\{i-1,i\right\}$ that can be connected to its designated receiver. Such assignment would lead the ratio $\frac{\tau_p(2)}{1-p} \rightarrow 2$ as $p \rightarrow 1$. In the following theorem, we analyze a transmission scheme based on this assignment.

\begin{thm}\label{thm:mtwoic}
For $M=2$, the following average per user DoF is achievable,
\begin{equation}\label{eq:mtwoic}
\tau_p(M=2) \geq \frac{1}{3} \left(1-p\right)\left(1 + \left(1-p\right)^3 + B.p\right),
\end{equation}
where,
\begin{eqnarray}
B &=& 3+\left(1+\left(1-p\right)^3\right)\left(1-\left(1-p\right)^2+p\left(1-p\right)^3\right)\nonumber\\ &+& p \left(1+(1-p)^2 \right),
\end{eqnarray}
and, 
\begin{equation}\label{eq:mtwoiclimit}
\lim_{p \rightarrow 1} \frac{\tau_p(2)}{1-p} = 2.
\end{equation}
\end{thm}
\begin{proof}
For any message assignment, no message can be transmitted if the links from both transmitters carrying the message to its designated receiver are absent, and hence, the average DoF achieved for each message is at most $1-p^2$. It follows that $\lim_{p\rightarrow 1} \frac{\tau_p(2)}{1-p} \leq \lim_{p \rightarrow 1} \frac{(1-p)(1+p)}{1-p} = 2$. We then need only to prove that the inner bound in~\eqref{eq:mtwoic} is valid. In the achieving scheme, each message is assigned to the two transmitters that may be connected to its designated receiver, i.e., ${\cal T}_i = \{i-1,i\}, \forall i\in[K]$. Also, in each network realization, each transmitter will transmit at most one message and any transmitted message will be received at its designated receiver without interference. It follows that 1 DoF is achieved for any message that is transmitted, and hence, the probability of transmission is the same as the average DoF achieved for each message.

Each message $W_i$ such that $i \equiv 0 \text{ mod } 3$ is transmitted through $X_{i-1}$ whenever $H_{i,i-1} \neq 0$, and is transmitted through $X_i$ whenever $H_{i,i-1}=0$ and $H_{i,i} \neq 0$. It follows that $d_0$ DoF is achieved for each of these messages, where,
\begin{equation}
d_0 = (1-p)(1+p).
\end{equation}

We now consider messages $W_i$ such that $i \equiv 1 \text{ mod } 3$. Any such message is transmitted through $X_{i-1}$ whenever $H_{i,i-1} \neq 0$ and $H_{i-1,i-1} = 0$. We note that whenever the channel coefficient $H_{i-1,i-1} \neq 0$, message $W_i$ cannot be transmitted through $X_{i-1}$ as the transmission of $W_i$ through $X_{i-1}$ in this case will prevent $W_{i-1}$ from being transmitted due to either interference at $Y_{i-1}$ or sharing the transmitter $X_{i-1}$. It follows that $d_1^{(1)} = p(1-p)$ DoF is achieved for transmission of $W_i$ through $X_{i-1}$. Also, message $W_i$ is transmitted through $X_i$ whenever it is not transmitted through $X_{i-1}$ and $H_{i,i} \neq 0$ and either $H_{i,i-1}=0$ or message $W_{i-1}$ is transmitted through $X_{i-2}$. More precisely, $W_i$ is transmitted through $X_i$ whenever all the following is true:
$H_{i,i} \neq 0$, and either $H_{i,i-1} =0$ or it is the case that $H_{i,i-1}\neq 0$ and $H_{i-1,i-1} \neq 0$ and $H_{i-1,i-2} \neq 0$.
It follows that $d_1^{(2)}=p\left(1-p\right)+\left(1-p\right)^4$ is achieved for transmission of $W_i$ through $X_{i}$, and hence, $d_1$ DoF is achieved for each message $W_i$ such that $i \equiv 1 \text{ mod } 3$, where,
\begin{equation}
d_1= d_1^{(1)} + d_1^{(2)} = 2p\left(1-p\right) + \left(1-p\right)^4.
\end{equation}

We now consider messages $W_i$ such that $i \equiv 2 \text{ mod } 3$. Any such message is transmitted through $X_{i-1}$ whenever all the following is true: 

\begin{itemize}
\item $H_{i,i-1} \neq 0$.
\item Either $H_{i-1,i-1}=0,$ or $W_{i-1}$ is not transmitted.
\item $W_{i+1}$ is not causing interference at $Y_i$.
\end{itemize} 

The first condition is satisfied with probability $(1-p)$. In order to compute the probability of satisfying the second condition, we note that $W_{i-1}$ is not transmitted for the case when $H_{i-1,i-1} \neq 0$ only if $W_{i-2}$ is transmitted through $X_{i-2}$ and causing interference at $Y_{i-1}$, i.e., only if $H_{i-2,i-3}=0$ and $H_{i-2,i-2} \neq 0$ and $H_{i-1,i-2} \neq 0$. It follows that the second condition is satisfied with probability $p + p(1-p)^3$. The third condition is not satisfied only if $H_{i,i} \neq 0$ and $H_{i+1,i} \neq 0$, and hence, will be satisfied with probability at least $1-\left(1-p\right)^2$. Moreover, even if if $H_{i,i} \neq 0$ and $H_{i+1,i} \neq 0$, the third condition can be satisfied if message $W_{i+1}$ can be transmitted through $X_{i+1}$ without causing interference at $Y_{i+2}$, i.e., if $H_{i+1,i+1} \neq 0$ and $H_{i+2,i+1}=0$. It follows that the third condition will be satisfied with probability $1-(1-p)^2+p(1-p)^3$, and $d_2^{(1)}$ DoF is achieved by transmission of $W_i$ through $X_{i-1}$, where,
\begin{equation}
d_2^{(1)} = p\left(1-p\right)\left(1+\left(1-p\right)^3\right)\left(1-\left(1-p\right)^2+p\left(1-p\right)^3\right).
\end{equation}
 
Message $W_i$ such that $i \equiv 2 \text{ mod } 3$ is transmitted through $X_i$ whenever $H_{i,i} \neq 0$, and $H_{i+1,i} = 0$, and either $H_{i,i-1} = 0$ or $W_{i-1}$ is transmitted through $X_{i-2}$. It follows that $d_2^{(2)}$ DoF is achieved by transmission of $W_i$ through $X_i$, where,
\begin{eqnarray}
d_2^{(2)} &=& p\left(1-p\right)\left(p + d_1^{(1)} \left(1-p\right) \right)
\\&=& p^2\left(1-p\right) \left(1+(1-p)^2 \right),
\end{eqnarray}
and hence, $d_2= d_2^{(1)} + d_2^{(2)}$ DoF is achieved for each message $W_i$ such that $i \equiv 2 \text{ mod } 3$. We finally get,
\begin{equation}
\tau_p(2) \geq \frac{d_0 + d_1 + d_2}{3},
\end{equation}   
which is the same inequality as in~\eqref{eq:mtwoic}.
\end{proof}

\begin{figure}
  \centering
\subfloat[]{\label{fig:mtwodof}\includegraphics[height=0.155\textwidth]{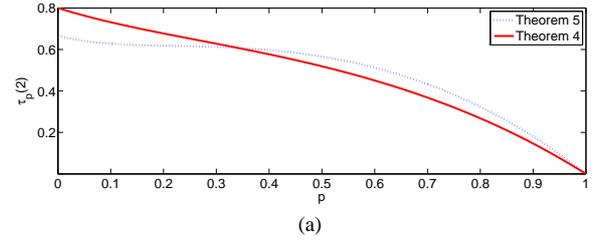}}                
\quad\quad\quad\quad\subfloat[]{\label{fig:mtwonormdof}\includegraphics[width=0.5\textwidth]{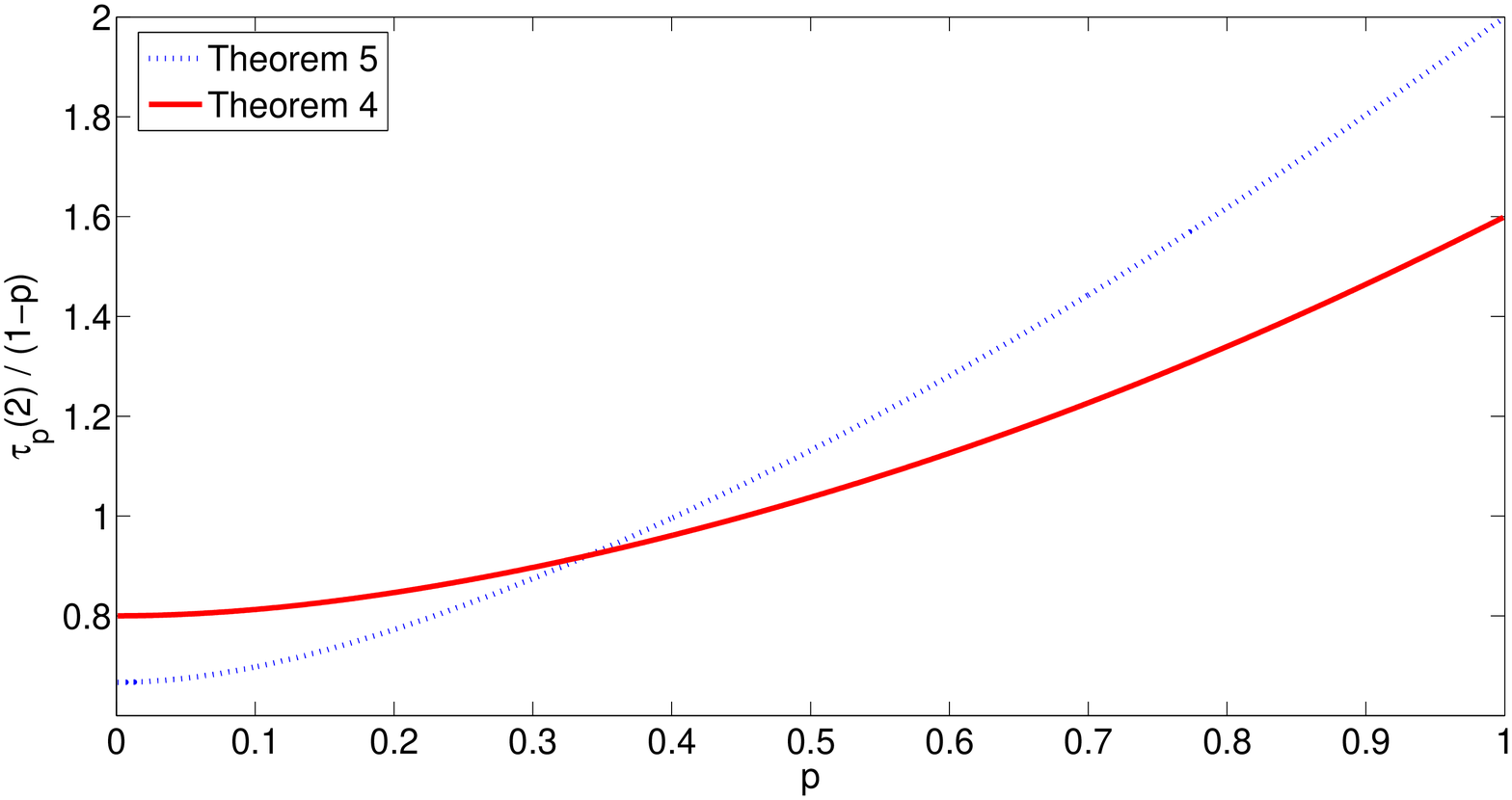}}
  \caption{Achieved inner bounds in Theorems $4$ and $5$. In $(a)$ we plot the achieved per user DoF. In $(b)$, we plot the achieved per user DoF normalized by $(1-p)$.}
  \label{fig:mtwo}
\end{figure}
We plot the inner bounds of~\eqref{eq:mtwoicaware} and~\eqref{eq:mtwoic} in Figure~\ref{fig:mtwo}. We note that below a threshold erasure probability $p \approx 0.34$, the scheme of Theorem~\ref{thm:mtwoicaware} is better, and hence  is proposed to be used in this case. For higher probabilities of erasure, the scheme of Theorem~\ref{thm:mtwoic} should be used. 
It is worth mentioning that we also studied a scheme based on the message assignment ${\cal T}_i = \{i,i+1\}, \forall i\in[K-1]$, that is introduced in~\cite{Lapidoth-Shamai-Wigger-ISIT07}. However, we did not include it here as it does not increase the maximum of the bounds derived in~\eqref{eq:mtwoicaware} and~\eqref{eq:mtwoic} at any value of $p$. Finally, although the considered channel model allows for using the interference alignment scheme of~\cite{Cadambe-IA} over multiple channel realizations (symbol extensions), all the proposed schemes require only coding over one channel realization because of the sparsity of the linear network.

\section{Conclusion}\label{sec:conclusion}
We considered the problem of assigning messages to transmitters in a linear interference network with link erasure probability $p$, under a constraint that limits the number of transmitters $M$ at which each message can be available. For the case where $M=1$, we identified the optimal message assignment strategies at different values of $p$, and characterized the average per user DoF $\tau_p(M=1)$. For general values of $M\geq 1$, we proved that there is no message assignment strategy that is optimal for all values of $p$. We finally introduced message assignment strategies for the case where $M=2$, and derived inner bounds on $\tau_p(M=2)$ that are asymptotically optimal as $p \rightarrow 0$ and as $p \rightarrow 1$.

\bibliographystyle{IEEEtran}

\end{document}